\documentclass[a4paper,12pt]{article}
\usepackage[dvipdfmx]{graphicx}
\usepackage{amssymb,amsmath}
\usepackage{amsthm}
\usepackage{bm}%
\usepackage{stmaryrd}
\usepackage{mathrsfs} 
\usepackage{cite}

\def\beq{\begin{equation}}
\def\eeq{\end{equation}}
\def\bea{\begin{eqnarray}}
\def\eea{\end{eqnarray}}
\def\nn{\nonumber}

\def\Z2{\mathbb{Z}_2^2}
\def\Zn{\mathbb{Z}_2^n}
\def\tg#1{\tilde{\gamma}_{#1}}
\def\g{\mathfrak{g}(n)}
\def\sQ{\mathsf{Q}}
\def\sH{\mathsf{H}}

\newtheorem{proposition}{Proposition}
\newtheorem{remark}{Remark}

\title{$\mathbb{Z}_2^n$-Graded extensions of supersymmetric quantum mechanics via Clifford algebras}

\author{N. Aizawa\thanks{{E-mail: {\em aizawa@p.s.osakafu-u.ac.jp}}}, K. Amakawa, S. Doi
\\[10pt]
Department of Physical Science, Osaka Prefecture University, \\
Nakamozu Campus, Sakai, Osaka 599-8531, Japan}
\begin{document}
\maketitle
\thispagestyle{empty}

\vfill
\begin{abstract}
It is shown that  the ${\cal N}=1$ supersymmetric quantum mechanics (SQM) can be extended to a $\mathbb{Z}_2^n$-graded superalgebra.  
This is done by presenting quantum mechanical models which realize, with the aid of Clifford gamma matrices, the $\Zn$-graded Poincar\'e algebra in one-dimensional spacetime. 
Reflecting the fact that the $\Zn$-graded Poincar\'e algebra has a number of central elements, 
a sequence of models defining the $\Zn$-graded version of SQM is provided for a given value of $n.$ 
In a model of the sequence, the central elements having the same $\Zn$-degree are realized as dependent or independent operators. 
It is observed that as use  the Clifford algebra of larger dimension, more central elements are realized as independent operators.  

\end{abstract}

\clearpage
\setcounter{page}{1}

\section{Introduction}

 In the recent paper \cite{BruDup}, Bruce and Duplij introduced a $\mathbb{Z}_2^2$-graded supersymmetric quantum mechanics (SQM) where $ \mathbb{Z}_2^2$ is an abbreviation of $ \mathbb{Z}_2\times \mathbb{Z}_2.$  
 It is a model of quantum mechanical Hamiltonian $H$ which can be factorized in two distinct ways, i.e., 
$ H = Q_{01}^2 = Q_{10}^2.$ 
The difference from the standard ${\cal N}=2$ SQM is that two supercharges $Q_{01}$ and $ Q_{10}$ have different degree. It  means that they close in commutator, instead of anticommutator, in a central element $Z.$ More explicitly they satisfy the relation $[Q_{01}, Q_{10}] =Z.$ 
Thus the algebra spanned by $ H, Q_{01}, Q_{10}$ and $Z$ is not a superalgebra, 
but a $\Z2$-graded superalgebra \cite{Ree,rw1,rw2,sch}. 
Each subset $ H, Q_{01}$ and $ H, Q_{10}$ forms a standard SQM. 
Thus the whole system of the $\mathbb{Z}_2^2$-graded  SQM can be regarded as a doubling of ${\cal N}=1$ SQM closed in a $\Z2$-graded superalgebra. 
It is elucidated in \cite{BruDup} that this system is different from the generalizations of SQM found in the literature. 

The Hamiltonian of Bruce and Duplij is a four by four matrix differential operator. 
It is surprising (at least for the authors) that there exists a $ \mathbb{Z}_2^2$-graded superalgebra behind such a simple Hamiltonian. 
It is then natural to anticipate that a generalization of $ \mathbb{Z}_2^2$  to $ \Zn$ will be possible. 
The purpose of the present work is to demonstrate that this is indeed the case. 
We present various models of $\Zn$-graded SQM in the subsequent sections. 
This is done by constructing a map from the standard SQM to $\Zn$-graded version by the use of the Clifford algebras of various dimensions. 
The present authors realized that the  $\mathbb{Z}_2^2$-graded SQM of Bruce and Duplij is obtained by a mapping from a Lie superalgebra to a $\mathbb{Z}_2^2$-graded superalgebra by the use of the Clifford algebra $ Cl(2)$ \cite{NAKASD}. 
So the results presented in this work is a kind of generalization of \cite{NAKASD}. 
The generalization is, however,  highly non-trivial problem in the following sense.  

The $\Zn$-graded SQM is defined as a quantum mechanical realization of the $\Zn$-graded extension of super-Poincar\'e algebra introduced by Bruce \cite{Bruce} (we make a dimensional reduction of the algebra in \cite{Bruce} which is defined in 4D Minkowski spacetime to 1D). 
Algebra of Bruce has central elements and the number of the central elements increases rapidly as $n$ becomes  larger (see \S \ref{SEC:Poin}).  
This allows us various possibilities of realizing the algebra. 
In some realization all or a part of the central elements having the same $\Zn$ degree are \textit{not} independent and in other realization all the central elements are linearly independent.

The present work is motivated by recent renewed interest in $\Zn$-graded superalgebras.  
They are also referred to as color superalgebras in the literature. 
However, these terminologies are bit confusing and misleading because the algebras under consideration  is, in mathematically more standard terminology, a graded Lie algebra by an abelian group. 
What makes the situation difficult is the fact that even fixing a group for grading possible graded Lie algebras are not unique \cite{rw1,rw2}. 
We shall follow the terminology ``$\Zn$-graded superalgebras", but simply say ``$\Zn$-graded Poincar\'e algebra" for $\Zn$-graded version of supersymmetry algebra of \cite{Bruce}. 

Mathematical interests of $\Zn$-graded superalgebras have been kept since their introduction and algebraic and geometric aspects and `higher graded' version of  supermanifolds have been continuously studied till today. 
Readers may refer the references, e.g. in \cite{Bruce,NaPsiJs2}.

On the other hand, physical implication of $\Zn$-graded superalgebras is not  clear even today. 
There exist some early works discussing on possible connection between $\Zn$-graded superalgebras and various physical problems such as supergravity, string theory, quasispin formalism and so on   \cite{LuRi,vas,jyw,zhe,LR,Toro1,Toro2}. 
Recently, $\Zn$-graded superalgebras started appearing again in physics literatures.  
We mention the works on modifying the spacetime symmetry \cite{tol2,Bruce} and works in connection with parastatistics \cite{tol,StoVDJ}. 
It is also remarkable that $ \Z2$-graded superalgebras appear as symmetries of partial differential equations. 
A very interesting and  physically important example is  that symmetries of wave  equations of nonrelativistic quantum mechanics for fermions are given by a $\Z2$-graded extension of the Schr\"odinger algebra \cite{AKTT1,AKTT2}. 
Furthermore, it is suggested in \cite{BruIba} that mixed symmetry tensors over  $\mathbb{Z}_2^n$-graded manifolds will have some connection to the double field theory since $\mathbb{Z}_2^n$-graded manifolds contain commuting fermions. 

 These works suggest that there exist more places in physics where $\Zn$-graded superalgebras play significant roles. 
The results of the present work, which show simple matrix differential operators have dynamical $\Zn$-graded symmetry, will give an another example of an intimate relation between physics and $\Zn$-graded superalgebras.

The plan of this paper is as follows: 
\S \ref{SEC:Prel} is a preliminary section where we collect basics used in the present work. 
The definition of the $\Zn$-graded Poincar\'e algebra and complex representations of the Clifford algebra $Cl(2n)$ are given. We start presenting models of $\Zn$-graded SQM from \S \ref{SEC:Mini}. 
We discuss a model of minimal dimension which is obtained by using $ Cl(2(n-1)) $ in \S \ref{SEC:Mini}. 
In this model  all the central elements having the same $\Zn$ degree are not independent. 
We also examine the spectrum of the Hamiltonian. 
In \S \ref{SEC:next}, a model  constructed by making use of $ Cl(2n)$ is presented in which degeneracy of the central elements is partially resolved. 
The model presented in \S \ref{SEC:max} is of maximal dimension in the sense that sequential action of any supercharges on a given state always produce a new (i.e.,  independent) state. 
In \S \ref{SEC:InBet}, we discuss the existence of a sequence of models from the minimal to  the maximal dimensions by presenting $n=4, 5 $ as examples. 
We summarize the results and give some remarks in \S \ref{SEC:CR}. 

%

\section{Preliminaries} \label{SEC:Prel}

\subsection{$\Zn$-graded Poincar\'e algebra in one-dimension}
\label{SEC:Poin}

In \cite{Bruce}, the super-Poincar\'e algebra in 4D Minkowski spacetime is generalized to $\Zn$-graded setting. By reducing the spacetime dimension to one, we have a $ \Zn$-graded  version of supersymmetry algebra ($\Zn$-graded Poincar\'e algebra for short) which gives a basis of $\Zn$-graded version of SQM. 
In this section, we review the $\Zn$-graded Poincar\'e algebra of ${\cal N} =1.$
The algebra is denoted simply by $\g.$ 

Let $ \bm{a} = (a_1, a_2, \dots, a_n) \in \Zn$ and define the \textit{parity} of $\bm{a}$ 
by $ \displaystyle p(\bm{a}) = \sum_{k=1}^n a_k \ (\mathrm{mod} \ 2). $ 
The $ \Zn$-graded Poincar\'e algebra $\g$ is spanned by the time translation (Hamiltonian) $H, $ supercharges $ Q_{\bm{a}} $ and central elements $ Z_{\bm{a}\bm{b}}. $ 
We assign the $\Zn$ degree to the elements of $\g$ as follows:
\begin{align}
   \deg(H) &= (0,0,\dots,0), \qquad
   \deg(Q_{\bm{a}}) = \bm{a}, \quad  p(\bm{a}) = 1,
   \nn \\
   \deg(Z_{\bm{a}\bm{b}}) &= \bm{a}+\bm{b} \neq (0,\dots,0), 
   \quad \bm{a} \neq \bm{b},
   \quad p(\bm{a}+\bm{b}) =0.
\end{align}
All supercharges have parity 1 and all central elements have parity 0. 
The defining relations of $\g$ are given by the $\Zn$-graded bracket 
$ \llbracket \cdot, \cdot \rrbracket : \g \times \g \to \g $ with the symmetry property
\begin{equation}
  \llbracket X_{\bm{a}}, Y_{\bm{b}} \rrbracket = -(-1)^{\bm{a}\cdot \bm{b}} \llbracket Y_{\bm{b}}, X_{\bm{a}} \rrbracket
\end{equation}
where $ \bm{a}\cdot \bm{b} $ is the inner product of two $n$-dimensional vectors $ \bm{a} $ and $\bm{b}:$
\begin{equation}
   \bm{a}\cdot \bm{b} = \sum_{j=1}^n a_j b_j \mod 2
\end{equation}
We remark that $ \bm{a}\cdot \bm{b} $ is computed in $\mathrm{mod} \ 2$ throughout this article. 
The $\Zn$-graded bracket is realized by commutator or anticommutator
\begin{equation}
 \llbracket X_{\bm{a}}, Y_{\bm{b}} \rrbracket =  X_{\bm{a}} Y_{\bm{b}} - (-1)^{\bm{a}\cdot \bm{b}}
Y_{\bm{b}} X_{\bm{a}}, \label{gradedcom}
\end{equation}
With the $\Zn$-graded bracket, the non-vanishing relations of $\g$ are summarized in a single expression:
\begin{equation}
   \llbracket Q_{\bm a}, Q_{\bm b} \rrbracket = 
   2\delta_{\bm{a},\bm{b}} H + 2i^{1-\bm{a}\cdot\bm{b}}\, Z_{\bm{a}\bm{b}}
   \label{DefRels}
\end{equation}
where $ \delta_{\bm{a},\bm{b}} = \prod \delta_{a_j,b_j}.$ 
We note that 
\begin{equation}
   Z_{\bm{a}\bm{b}} = -(-1)^{\bm{a}\cdot \bm{b}} Z_{\bm{b}\bm{a}} \ \mathrm{for} \ \bm{a} \neq \bm{b}, 
   \qquad 
   Z_{\bm{a}\bm{a}} = 0.
\end{equation}

The order of the abelian group $ \Zn$ is given by $ |\Zn| = 2^n.$ 
The number of parity 0 elements is  equal to that of  parity 1 elements and it is $2^{n-1}.$ 
The algebra $\g$ is a direct sum (as a vector space) of $2^n$ subspaces and each of them is  labelled by an element of $\Zn:$ $ \g = \bigoplus_{\bm{a} \in \Zn} \g_{\bm{a}}. $ 
The subspace $ \g_{\bm{a}}$ of $ p(\bm{a}) = 1 $ is spanned by only one supercharge $ Q_{\bm{a}}$ (this is the reason why we refer  $\g$  to as ${\cal N}=1$) so that total number of supercharges is $ 2^{n-1}. $ 
The Hamiltonian $H$ spans the one-dimensional subspace $ \g_{(0,0,\dots,0)}$ and $ Z_{\bm{a}\bm{b}}$ spans the other subspaces of parity 0. 
From \eqref{DefRels} we see that the total number of the central elements is equal to 
the number of combinations taking two from $ 2^{n-1},$ that is, 
  $  2^{n-2}(2^{n-1}-1). $ 
The value of this number  increases rapidly as the grading group $\Zn$ becomes larger, i.e., $n$ goes bigger.  
Thus the subspaces spanned by the central elements are, in general, not one-dimensional. 
The dimension of the subspaces is $ 2^{n-2}$ which is obtained by dividing the total number of central elements by the number of subspaces of parity 0 except $ \g_{(0,0,\dots,0)}. $ 
We present some examples of the total number of supercharges $\#(Q),$ central elements $\#(Z)$ and the dimension $\dim(Z)$ of the subspace spanned by the central elements:
\[
  \begin{array}{c|ccccccccc}
     n       &  2  &  3  &  4  &  5  &  6  & 7  &  8  & 9 & 10 \\ \hline
    \#(Q)    &  2  &  4  &  8  &  16 & 32  & 64 & 128 & 256 & 512 \\
    \#(Z)    &  1  &  6  & 28  & 120 & 496 & 2016 & 8128 & 32640 & 130816 \\
    \dim(Z)  &  1  &  2  &  4  &  8  & 16  &  32  & 64   & 128   &  256
  \end{array}
\]

To show the structure of $\g$ more clearly, we take $ \mathfrak{g}(3)$ as an example.  
The $ \mathbb{Z}_2^3$-graded Poincar\'e algebra $ \mathfrak{g}(3)$ has four supercharges $ Q_{100}, Q_{010}, Q_{001}, Q_{111} $ and 
six central elements which span parity 0 subspaces:
\begin{align}
   \mathfrak{g}(3)_{(1,1,0)} &= \langle \ Z_{100,010},\ Z_{001,111} \ \rangle,
    \nn \\
   \mathfrak{g}(3)_{(1,0,1)} &= \langle \ Z_{100,001},\ Z_{010,111} \ \rangle,
     \nn \\
   \mathfrak{g}(3)_{(0,1,1)} &= \langle \ Z_{010,001},\ Z_{100,111} \ \rangle.
   \label{Z3evensubsp}
\end{align}
where the notation of $ \bm{a} \in \mathbb{Z}_2^3$ attached to the elements of $\mathfrak{g}(3)$ is changed slightly for better readablity. 
The defining relations of $\mathfrak{g}(3)$ are summarized below. 
They are indicated by the commutator $[ \cdot, \cdot ] $ and the anticommutator $\{ \cdot, \cdot \}.$

For $ \bm{a}, \bm{b} \in \{ \; (1,0,0), (0,1,0), (0,0,1), (1,1,1)  \;\} $
\begin{equation}
  \{Q_{\bm{a}}, Q_{\bm{a}} \} = 2H,  \qquad 
  [ H, Q_{\bm{a}}] = [H, Z_{\bm{a}\bm{b}}] = 0,
\end{equation}
and for $ \bm{a}, \bm{b} \in \{ \; (1,0,0), (0,1,0), (0,0,1) \;\} $
\begin{equation}
    [Q_{\bm{a}}, Q_{\bm{b}}] = 2i Z_{\bm{a}\bm{b}}, \qquad 
  \{Q_{\bm{a}}, Q_{111} \} = 2Z_{\bm{a},111}.
\end{equation}
The centrality of $Z_{\bm{a}\bm{b}}$ means that the following commutators/anticommutators  vanish identically. For simplicity, notations such as $\mathfrak{g}_{110} := \mathfrak{g}(3)_{(1,1,0)} $ are used for the subspaces in \eqref{Z3evensubsp}
\begin{alignat}{3}
  & \{Q_{100},\, \mathfrak{g}_{110} \}, & \qquad
  & \{Q_{100},\, \mathfrak{g}_{101} \}, & \qquad
  & [ Q_{100},\, \mathfrak{g}_{011} ],
  \nn \\
  & \{Q_{010},\, \mathfrak{g}_{110} \}, & \qquad
  & \{Q_{010},\, \mathfrak{g}_{011} \}, & \qquad
  &  [Q_{010},\, \mathfrak{g}_{101} ], 
  \nn \\
  & \{Q_{001},\, \mathfrak{g}_{101} \}, & \qquad
  & \{Q_{001},\, \mathfrak{g}_{011} \}, & \qquad
  &  [Q_{001},\, \mathfrak{g}_{110} ],   
  \nn \\
  & [Q_{111},\, \mathfrak{g}_{\bm{a}}], & \bm{a} &=110, 101, 011
  \nn \\
  & [ \mathfrak{g}_{110}, \, \mathfrak{g}_{110} ], & 
  & \{ \mathfrak{g}_{110}, \, \mathfrak{g}_{101} \}, & 
  & \{ \mathfrak{g}_{110}, \, \mathfrak{g}_{011} \},
  \nn \\
  & [\mathfrak{g}_{101}, \, \mathfrak{g}_{101} ],  &
  & \{ \mathfrak{g}_{101}, \, \mathfrak{g}_{011} \}, &
  & [\mathfrak{g}_{011}, \, \mathfrak{g}_{011}].
\end{alignat}

\subsection{$\Zn$-graded SQM as representations of $\Zn$-graded Poincar\'e algebra}

The $\Zn$-graded SQM is a quantum mechanical realization of $ \g $ introduced in \S \ref{SEC:Poin}. 
It means that we introduce a Hilbert space having $\Zn$-grading structure
\begin{equation}
  \mathscr{H} = \bigoplus_{\bm{a} \in \Zn} \mathscr{H}_{\bm{a}}
\end{equation}
and the elements of $\g$ are realized as matrix differential operators acting on $\mathscr{H}.$ 
The supercharges are first order differential operators and $H$ is a diagonal matrix whose entries have the form of quantum mechanical Hamiltonian (2nd order differential operator). 
As the Hilbert space $\mathscr{H}$ consists of $|\Zn| = 2^n$ subspaces, 
the minimal size of the matrix differential operators should be $2^n.$ 

In the subsequent sections we show the existence of $\Zn$-graded SQM given by various sizes of matrix differential operators. We refer the size of matrix to as \textit{dimension} of $\Zn$-graded SQM. 
Our method constructing the $\Zn$-graded SQM is to find a mapping from the ordinary SQM to $\Zn$-graded counterpart. This is done by using the Clifford algebra and the choice of different Clifford algebra produces $\Zn$-graded SQM of different dimension. 
One may find various mappings from Lie superalgebras to $\Zn$-graded superalgebras in literature \cite{rw1,rw2,sch,NA2,NAKASD} which also using the Clifford algebras (especially mappings giving $ \mathbb{Z}_2^2$-graded superalgebras). 
The mappings presented in the present work are novel ones which are suitable to consider quantum mechanical realization.  

 Now we recall the building blocks of our construction, i.e., 
${\cal N}=1$ SQM and irreducible representations of the Clifford algebra $ Cl(2n).$  
The ${\cal N}=1$ SQM is given by $2\times 2$ matrix differential operators:
\begin{equation}
    \mathsf{Q} = \frac{1}{\sqrt{2}}(\sigma_2 p + \sigma_1 W(x)), \qquad
    \mathsf{H} = \frac{1}{2}(p^2+W(x)^2) \mathbb{I}_2 - \frac{1}{2}W'(x) \sigma_3
    \label{StandardSQM}
\end{equation}
where $ p = -i\partial_x$ and $ W' = \frac{d W(x)}{dx}.$ 
The symbols $\sigma_i$ and $ \mathbb{I}_m$ denote the  Pauli matrix and   the $m \times m$ identity matrix, respectively. 
The operators are chosen to be hermitian. 
They satisfy the relations:
\begin{equation}
   \{ \sQ, \sQ \} = 2\sH, \qquad [\mathsf{H}, \mathsf{Q}] = 0 \label{StandSQM}
\end{equation}
which are the defining relations of supersymmetry algebra (or $\mathbb{Z}_2$-graded Poincar\'e algebra) in one dimension. 
We also use the following expressions of \eqref{StandardSQM}
\begin{equation}
   \mathsf{Q} = \begin{pmatrix}
                  0 & A^{\dagger} \\ A & 0
                \end{pmatrix},
   \qquad
   \mathsf{H} = \begin{pmatrix}
                  A^{\dagger}A & 0 \\ 0 & A A^{\dagger}
                \end{pmatrix}
   \label{matrixSQM}
\end{equation}
where
\begin{equation}
    A = \frac{1}{\sqrt{2}} (ip + W(x)), 
    \qquad
    A^{\dagger} = \frac{1}{\sqrt{2}} (-ip + W(x)).
    \label{ladderOP}
\end{equation}

  The Clifford algebra $Cl(2n)$ is generated by $\gamma_j\ (j= 1,2,\dots,2n)$  subject to the conditions
\begin{equation}
   \{ \gamma_j, \gamma_k \}  = 2 \delta_{jk}.
\end{equation}
The hermitian complex irreducible representation of $Cl(2n)$ is $2^n$ dimensional and 
given explicitly as follows \cite{Okubo,CaRoTo}:
\begin{align}
  \gamma_1 &= \sigma_1^{\otimes n},
    \qquad
  \gamma_j = \sigma_1^{\otimes (n-j+1)} \otimes \sigma_3 \otimes \mathbb{I}_2^{\otimes (j-2)}, \quad 2 \leq j \leq n,
    \nn \\
   \tilde{\gamma}_j &:= \gamma_{j+n} = \sigma_1^{\otimes(n-j)} \otimes \sigma_2 \otimes \mathbb{I}_2^{\otimes(j-1)}, 
    \quad 1 \leq j \leq n
    \label{Clrep}
\end{align}
Note that $ \gamma_j $ and $\tg{j}$ have non-zero entries at the same position in a $2^n \times 2^n$ matrix. 
The product $ i \gamma_j \tg{j} $ is diagonal, hermitian and idempotent:
\begin{alignat}{2}
   i \gamma_1 \tg{1} &= -\mathbb{I}_2^{\otimes (n-1)} \otimes \sigma_3,
   \nonumber \\
    i \gamma_j \tg{j} &= -\mathbb{I}_2^{\otimes(n-j)}\otimes \sigma_3 \otimes \sigma_3 \otimes \mathbb{I}_2^{(j-2)}, & &\quad 2 \leq j \leq n
    \nonumber 
    \\
   (i \gamma_j \tg{j})^{\dagger} &= i \gamma_j \tg{j}, \qquad
   (i \gamma_j \tg{j})^2 = \mathbb{I}_{2^n}, 
   & & \quad 1 \leq j \leq n     
   \label{maxrel1} 
\end{alignat}
Furthermore, it is easy to see the relations
\begin{equation}
  \{ \gamma_j,\, i\gamma_j \tg{j} \} = 
  [\gamma_j, \, i \gamma_k \tg{k} ] = 
  [i\gamma_j \tg{j},\, i\gamma_k\tg{k}] = 0, \quad j \neq k 
  \label{maxrel2}
\end{equation}
These relations will be used to construct models of the $\Zn$-graded SQM.

Before giving the models of $\Zn$-graded SQM, we point out a similarity to the standard SQM. 
In the quantum mechanical realization of $\g$, the supercharge $ Q_{\bm{a}}$ is realized as a hermitian operator and it factorizes the Hamiltonian: $ H = Q_{\bm{a}}^2. $ 
Thus one may repeat the same discussion as the standard SQM and see the following (see e.g. \cite{CoKhSu,Junker}):
\begin{proposition}
 In any models of the $\Zn$-graded SQM, any eigenvalue $E$ of the Hamiltonian $H$  is positive semi-definite $ E \geq 0.$  
\end{proposition}

%
\setcounter{equation}{0}
\section{Minimal dimensional $\Zn$-graded SQM by $ Cl(2(n-1)) $}
\label{SEC:Mini}

Let us consider the Clifford algebra $Cl(2(n-1)) $
whose representation \eqref{Clrep} is  $2^{n-1}$ dimensional.  

\begin{proposition} \label{P2}
Define the $2^n$ dimensional matrix differential operators $H, Q_{\bm{a}}, Z_{\bm{a}\bm{b}}\; (\bm{a} \neq \bm{b})$ by
\begin{align}
   H &= \mathbb{I}_{2^{n-1}} \otimes \sH, 
   \qquad
   Q_{\bm{a}} = 
        \begin{cases}
              X_{\bm{a}} \otimes i \sQ \sigma_3, & a_n = 0 
               \\[5pt]
              X_{\bm{a}} \otimes \sQ,          & a_n = 1
        \end{cases},
   \nn \\
   Z_{\bm{a}\bm{b}} &= 
        \begin{cases}
            (-i)^{1-\bm{a}\cdot\bm{b}} X_{\bm{a}} X_{\bm{b}} \otimes \mathsf{H}, & a_n = b_n
             \\[5pt]
            (-1)^{b_n}\, i^{\bm{a}\cdot\bm{b}} X_{\bm{a}} X_{\bm{b}} \otimes \mathsf{H} \sigma_3, & a_n \neq b_n
        \end{cases}
        \label{MinimalSQM}
\end{align}
where 
\begin{equation}
   X_{\bm{a}} = i^{h(\bm{a})} \gamma_1^{a_1} \gamma_2^{a_2} \cdots \gamma_{n-1}^{a_{n-1}},
   \qquad
   h(\bm{a}) = \sum_{j=1}^{n-2} \sum_{k=j+1}^{n-1} a_j a_k.
   \label{MinimalX}
\end{equation}
Then the matrix differential operators realize the relations \eqref{DefRels} of $\g.$ 
Thus they give a $2^n$ dimensional $\Zn$-graded SQM. 
\end{proposition}
\begin{proof}
 One may verify the relations in \eqref{DefRels} by direct computation using the properties of $ X_{\bm{a}}. $ 
First, note that $X_{\bm{a}}$ is hermitian and idempotent:
\begin{equation}
   X_{\bm{a}}^{\dagger} = X_{\bm{a}}, \qquad X_{\bm{a}}^2 = \mathbb{I}_{2^{n-1}}.
\end{equation}
It is then immediate to see, noting  $ (i \sQ \sigma_3)^2 = \mathsf{Q}^2 = \mathsf{H},$ that $ Q_{\bm{a}}^2 = X_{\bm{a}}^2 \otimes \mathsf{H} = H $ for any supercharges. 
Thus  we have verified that $ \{ Q_{\bm{a}}, Q_{\bm{a}} \} = 2H $ and $ [H, Q_{\bm{a}}] = 0.$ 

Next, we show the relation: 
\begin{equation}
   \llbracket Q_{\bm{a}}, Q_{\bm{b}} \rrbracket = 2i^{1-\bm{a} \cdot \bm{b}} Z_{\bm{a}\bm{b}}, \quad \bm{a} \neq \bm{b}
   \label{QQZ}
\end{equation}
To this end, four separate cases should be considered:
\begin{align}
 \bm{a}\cdot\bm{b} =0 \quad &\Rightarrow \quad  \llbracket Q_{\bm{a}}, Q_{\bm{b}} \rrbracket =[ Q_{\bm{a}}, Q_{\bm{b}} ] = 
   \begin{cases}
    [ X_{\bm{a}}, X_{\bm{b}} ] \otimes \mathsf{H} & a_n = b_n
     \\[9pt]
     -i\{ X_{\bm{a}}, X_{\bm{b}} \} \otimes \mathsf{H} \sigma_3, & a_n =0,\, b_n=1
   \end{cases} 
   \nn \\[10pt]
 \bm{a}\cdot\bm{b} =1 \quad &\Rightarrow \quad  \llbracket Q_{\bm{a}}, Q_{\bm{b}} \rrbracket = \{ Q_{\bm{a}}, Q_{\bm{b}} \} =    
   \begin{cases}
    \{ X_{\bm{a}}, X_{\bm{b}} \} \otimes \mathsf{H} & a_n = b_n
     \\[9pt]
     -i[ X_{\bm{a}}, X_{\bm{b}} ] \otimes \mathsf{H} \sigma_3, & a_n =0,\, b_n=1
   \end{cases}
\end{align}
We are able to choose  $ a_n = 0,\, b_n=1$ without loss of generality since the alternate choice cause only the overall sign change. 
The relation \eqref{QQZ} is a direct consequence of the fact that if $ \bm{a} \neq \bm{b},$ then $ X_{\bm{a}}$  commutes or anticommutes with $ X_{\bm{b}}:$ 
\begin{alignat}{2}
   [X_{\bm{a}}, X_{\bm{b}}] &= 0 &\quad & \mathrm{if \ }
     \begin{cases}
        \mathrm{(i)} \ \bm{a}\cdot\bm{b} =0, \ a_n \neq b_n
        \\[5pt]
        \mathrm{(ii)} \ \bm{a}\cdot\bm{b} =1, \ a_n = b_n
     \end{cases}
   \nn \\[7pt]
   \{X_{\bm{a}}, X_{\bm{b}} \} &= 0 &\quad & \mathrm{if \ }
     \begin{cases}
        \mathrm{(iii)} \ \bm{a}\cdot\bm{b} =0, \ a_n = b_n
        \\[5pt]
        \mathrm{(iv)} \ \bm{a}\cdot\bm{b} =1, \ a_n \neq b_n
     \end{cases} 
     \label{Xcommute}
\end{alignat}
The relation \eqref{Xcommute} is verified in the following way. 
One may write $ X_{\bm{a}} = X'_{\bm{a}} Y $ and $ X_{\bm{b}} = X'_{\bm{b}} Y $ where $Y$ denotes the product of $\gamma$'s common to $X_{\bm{a}}$ and $ X_{\bm{b}}.$ 

\medskip
\noindent
a)  $ a_n \neq b_n. $ 
This means that one of $ a_n, b_n $ is 0 and the another is 1. 
For clarity, we consider the case $ a_n = 0, b_n = 1.$ 
Then $ X_{\bm{a}}$ is a product of an odd number of $\gamma$'s while $ X_{\bm{b}}$ is a product of an even number of $\gamma$'s. 
Note that $ \bm{a} \cdot \bm{b} = \sum_{j=1}^{n-1} a_j b_j$ is a summation up to $n-1$ (not up to $n$) as $ a_n b_n = 0. $
Thus $ \bm{a} \cdot \bm{b} = 0 $ implies that $Y$ is a product of even number of $\gamma$'s so that $ Y$ commutes with $ X'_{\bm{a}} $ and $ X'_{\bm{b}}. $ 
Furthermore, $ X'_{\bm{a}} $ and $ X'_{\bm{b}} $ are a product of  odd and even number of $\gamma$'s, respectively. Thus we see $ [X'_{\bm{a}} ,X'_{\bm{b}} ] = 0. $ 
These observations conclude that $ [X_{\bm{a}}, X_{\bm{b}}] = 0.$ 
The case $ a_n = 1, b_n = 0$ is treated in the same way and the case (i) has been verified.  
When  $ \bm{a} \cdot \bm{b} = 1, $ $Y$ is a product of odd number of $\gamma$'s so that 
$ X'_{\bm{a}} $ and $ X'_{\bm{b}} $ are a product of even and odd number of $\gamma$'s, respectively. Thus $[ X'_{\bm{a}}, Y ] = \{X'_{\bm{b}}, Y\} = 0 $ and $ [X'_{\bm{a}}, X'_{\bm{b}}] = 0.$ It follows the case (iv). 

\medskip
\noindent
b) $ a_n = b_n = 0. $ 
For this value of $a_n, b_n$, $ X_{\bm{a}}$ and $X_{\bm{b}}$ are a product of an odd number of $\gamma$'s. 
Because of the reason same as the case a), 
$ Y, X'_{\bm{a}}$ and $ X'_{\bm{b}}$ consist of the following number of $\gamma$'s: 

\begin{center}
\begin{tabular}{ccc}
   $ \bm{a} \cdot \bm{b} $ & $ Y $ & $X'_{\bm{a}}, X'_{\bm{b}} $
   \\ \hline
      0 & even & odd 
   \\
      1 & odd & even
\end{tabular}
\end{center}

\noindent
Therefore, $ \bm{a} \cdot \bm{b} = 0 $ implies that $[X'_{\bm{a}}, Y] = [X'_{\bm{b}},Y] = 0 $ and $ \{ X'_{\bm{a}}, X'_{\bm{b}} \} = 0.$ 
If follows that $ \{ X_{\bm{a}}, X_{\bm{b}} \} = 0. $ 
While, $ \bm{a} \cdot \bm{b} = 1 $ implies that $ Y,  X'_{\bm{a}},  X'_{\bm{b}} $ commute each other so that $ [ X_{\bm{a}}, X_{\bm{b}} ] = 0. $

\medskip
\noindent
c) $ a_n = b_n = 1.$ 
For this value of $a_n, b_n$, $ X_{\bm{a}}$ and $X_{\bm{b}}$ are a product of an even number of $\gamma$'s. 
Noting that $ \bm{a} \cdot \bm{b} = \sum_{j=1}^{n-1} a_j b_j + 1,$ one may see that 
$ Y, X'_{\bm{a}}$ and $ X'_{\bm{b}}$ consist of the following number of $\gamma$'s: 

\begin{center}
\begin{tabular}{ccc}
   $ \bm{a} \cdot \bm{b} $ & $ Y $ & $X'_{\bm{a}}, X'_{\bm{b}} $
   \\ \hline
      0 & odd & odd 
   \\
      1 & even & even
\end{tabular}
\end{center}

\noindent
Therefore, $ \bm{a} \cdot \bm{b} = 0 $ implies that 
$ Y,  X'_{\bm{a}},  X'_{\bm{b}} $ anticommute each other so that $ \{ X_{\bm{a}}, X_{\bm{b}} \} = 0. $
While, $ \bm{a} \cdot \bm{b} = 1 $ implies that $ Y,  X'_{\bm{a}},  X'_{\bm{b}} $ commute each other so that $ [ X_{\bm{a}}, X_{\bm{b}} ] = 0. $

Thus the cases b) and c) have verified the cases (ii) and (iii).

Finally, we need to show the centrality of $ Z_{\bm{a}\bm{b}}. $ 
To this end, it is sufficient to show that $ \llbracket Z_{\bm{a}\bm{b}}, Q_{\bm{c}} \rrbracket =0 $ since $ \llbracket Z_{\bm{a}\bm{b}}, Z_{\bm{c}\bm{d}} \rrbracket = \llbracket Z_{\bm{a}\bm{b}}, H \rrbracket = 0 $  follows immediately from \eqref{QQZ}. 
Using $ [\mathsf{H}, \sigma_3] = \{ \mathsf{Q}, \sigma_3 \} = 0, $ one may see that 
if $ a_n = b_n $ then
\begin{equation}
  \llbracket Z_{\bm{a}\bm{b}}, Q_{\bm{c}} \rrbracket = \llbracket X_{\bm{a}} X_{\bm{b}},\, X_{\bm{c}} \rrbracket \otimes \mathsf{HQ} \sigma_3^{1-a_n}
\end{equation}
and if $ a_n \neq b_n $ then
\begin{align}
  [ Z_{\bm{a}\bm{b}}, Q_{\bm{c}} ] &= \{  X_{\bm{a}} X_{\bm{b}},\, X_{\bm{c}} \} \otimes 
  \mathsf{HQ} \sigma_3^{a_n},
  \nn \\
  \{ Z_{\bm{a}\bm{b}}, Q_{\bm{c}} \} &= [  X_{\bm{a}} X_{\bm{b}},\, X_{\bm{c}} ] \otimes  \mathsf{HQ} \sigma_3^{a_n}
\end{align}
where we have dropped the irrelevant numerical factors.
Now we show that $ \llbracket X_{\bm{a}} X_{\bm{b}},\, X_{\bm{c}} \rrbracket  = 0 $ holds true identically. 

\noindent
a) $ a_n = b_n $ and $ (\bm{a} + \bm{b})\cdot \bm{c} = 0.$ 
In this case $ \bm{a} \cdot \bm{c} = \bm{b} \cdot \bm{c}. $ 

If the both inner product is equal to 0, one may use (iii) of \eqref{Xcommute} then we see that
\begin{equation}
  [ X_{\bm{a}} X_{\bm{b}},\, X_{\bm{c}} ] = 
  X_{\bm{a}} \{ X_{\bm{b}},\, X_{\bm{c}} \} - \{ X_{\bm{a}},\,  X_{\bm{c}} \} X_{\bm{b}} = 0.
  \label{Com2Anti}
\end{equation}
If the both inner product is equal to 1, one may use (ii) of \eqref{Xcommute}
\begin{equation}
    [ X_{\bm{a}} X_{\bm{b}},\, X_{\bm{c}} ] = 
    X_{\bm{a}} [X_{\bm{b}},\, X_{\bm{c}} ] + [X_{\bm{a}} , \,X_{\bm{c}} ] X_{\bm{b}} = 0.
    \label{Com2Com}
\end{equation}

\noindent
b) $ a_n = b_n $ and $ (\bm{a} + \bm{b})\cdot \bm{c} = 1.$
There are two subcases. 

If $ \bm{a} \cdot \bm{c} = 0 $ and $ \bm{b} \cdot \bm{c} = 1, $ then because of (ii) (iii) of \eqref{Xcommute} we have
\begin{equation}
  \{  X_{\bm{a}} X_{\bm{b}},\, X_{\bm{c}} \} = X_{\bm{a}} [ X_{\bm{b}}, \, X_{\bm{c}}] 
  + \{ X_{\bm{a}} , \, X_{\bm{c}} \} X_{\bm{b}} = 0.   
  \label{AntiExpand1} 
\end{equation}
If $ \bm{a} \cdot \bm{c} = 1 $ and $ \bm{b} \cdot \bm{c} = 0, $ then similarly
\begin{equation}
  \{  X_{\bm{a}} X_{\bm{b}},\, X_{\bm{c}} \} = X_{\bm{a}} \{ X_{\bm{b}},\, X_{\bm{c}} \} 
  - [X_{\bm{a}},\, X_{\bm{c}}] X_{\bm{b}} = 0.
  \label{AntiExpand2}
\end{equation}

\noindent
c) $ a_n \neq b_n.$ 
Repeating the computations similar to a) and b), but using (i) (iv) of \eqref{Xcommute}, one may easily see $ \llbracket X_{\bm{a}} X_{\bm{b}},\, X_{\bm{c}} \rrbracket  = 0. $

This completes the proof. 
\end{proof}

\begin{remark} \label{R1}
In the representation \eqref{MinimalSQM}, all the central elements belonging to a given subspace $\g_{\bm{a}}$ differ only by a constant multiple, so $ \dim \g_{\bm{a}} = 1.$ 
This is readily seen by noting that $ X_{\bm{a}} X_{\bm{b}} \sim \gamma_1^{a_1+b_1} \gamma_2^{a_2+b_2} \cdots \gamma_{n-1}^{a_{n-1}+b_{n-1}}$ where $\sim$ denotes equality up to a constant multiple (we use this notation throughout this article). 
It follows that  
$   X_{\bm{a}} X_{\bm{b}} \sim  X_{\bm{c}} X_{\bm{d}} $ 
if $\bm{a}+\bm{b} = \bm{c}+\bm{d}$ which shows that the central elements in $\g_{\bm{a}+\bm{b}}$ differ only by a constant multiple.  
\end{remark}

\medskip
The operators in \eqref{MinimalSQM} act on the Hilbert space $ \mathscr{H} = L^2({\cal I}) \otimes \mathbb{C}^{2^n} $ where $ {\cal I} \subset \mathbb{R}$ is an open interval of $\mathbb{R}$ determined by the choice of potential. 
The action of $ Q_{\bm{a}}$ and $ Z_{\bm{a}\bm{b}}$ on the subspace $ \mathscr{H}_{\bm{a}}$ defines a mapping form one subspace to another with different $\Zn$-grading:
\begin{equation}
    Q_{\bm{a}} \mathscr{H}_{\bm{b}} \subset \mathscr{H}_{\bm{a}+\bm{b}},
    \qquad
    Z_{\bm{a}\bm{b}} \mathscr{H}_{\bm{c}} \subset \mathscr{H}_{\bm{a}+\bm{b}+\bm{c}}.
\end{equation}

Let us investigate the spectrum of the Hamiltonian $H.$ 
The zero energy ground state $ \Psi_0(x) \in \mathscr{H}$ is defined as the state  annihilated by all supercharges: $ Q_{\bm{a}} \Psi_0(x) = 0 $ for all $\bm{a}$ of parity 1.  
The action of any supercharge on $\Psi_0(x)$ yields the same conditions. 
To see this, we divide $2^n$ components of $\Psi_0(x)$ into two groups as follows:
\begin{equation}
   \Psi_0(x) = (\psi_1(x),\varphi_1(x),\psi_2(x),\varphi_2(x), \dots, \psi_{2^n-1}(x), \varphi_{2^n-1}(x))^T
\end{equation}
Members of one of the groups is denoted by $\psi_i $ and of the other group by $\varphi_i.$  
Because of   the form of $Q_{\bm{a}}$ in \eqref{MinimalSQM}, $ Q_{\bm{a}} \Psi_0(x) $ consists of the action of $ \mathsf{Q}$ or $ \mathsf{Q}\sigma_3$  on a pair $(\psi_i(x), \varphi_i(x)).$ 
Recalling the form of $\mathsf{Q} $ given in \eqref{matrixSQM}, we see that this action gives the relations:
\begin{equation}
       A \psi_i(x) = 0, \qquad A^{\dagger} \varphi_i(x) = 0. \label{GroundCond}
\end{equation}
If $ \mathscr{H} = L^2(\mathbb{R}) \otimes \mathbb{C}^{2^n}, $ it is known that $ \psi_i(x) $ and $ \varphi_i(x)$ satisfying \eqref{GroundCond} are not normalizable simultaneously \cite{CoKhSu,Junker}. 
Thus, only one of the sets $ \{ \; \psi_i(x) \; \}, \{ \; \varphi_i(x) \; \} $ is normalizable, or both sets are not normalizable. 
If the set $ \{ \; \psi_i(x) \; \} $ is normalizable, then 
$ \Psi_0(x) = (\psi_1(x),0,\psi_2(x),0, \dots, \psi_{2^n-1}(x), 0)^T \equiv \Psi_0^{(1)}  $ gives the ground state.  
While, if the other set $ \{ \; \varphi_i(x) \; \} $ is normalizable, then 
$ \Psi_0(x) = (0,\varphi_1(x),0,\varphi_2(x), \dots, 0, \varphi_{2^n-1}(x))^T \equiv \Psi_0^{(2)} $ gives the ground state.   
Obviously, $\Psi_0^{(1)} $ and $\Psi_0^{(2)} $ belong to the different subspaces of $\mathscr{H}.$  
In this way, we have proved the following:
\begin{proposition} \label{P3}
The zero energy state of $H$ in \eqref{MinimalSQM} is either non-existent or $2^{n-1}$ fold degenerate. The degenerate wave function is given by $ \Psi_0^{(1)}$ or $ \Psi_0^{(2)}. $
\end{proposition}

On the excited states, one may see the following:
\begin{proposition} \label{P4}
The excited states of $H$ in \eqref{MinimalSQM}  are $2^n$ fold degenerate.
\end{proposition}
\begin{proof} 
Let $ \bm{0} = (0,0,\dots,0) \in \Zn $ and 
\begin{equation}
   H \psi_{\bm{0}}(x) = E \psi_{\bm{0}}(x), \quad 
   E > 0, \quad \psi_{\bm{0}}(x) \in \mathscr{H}_{\bm{0}}
\end{equation}
where $ \psi_{\bm{0}}(x)$ is the $2^n$ components vector of the form
$
  \psi_{\bm{0}}(x) = (\psi(x), 0, \dots, 0)^T.
$ 
It follows from $ [H, Q_{\bm{a}}] = [H, Z_{\bm{a}\bm{b}}] = 0$ that the states $ Q_{\bm{a}} \psi_{\bm{0}} \in \mathscr{H}_{\bm{a}} $ and $ Z_{\bm{a}\bm{b}} \psi_{\bm{0}} \in \mathscr{H}_{\bm{a}+\bm{b}} $ are also eigenfunctions of $H$ of the same energy $E$. 

From \eqref{MinimalSQM} $ Q_{\bm{a}} Q_{\bm{b}} \sim Z_{\bm{a}\bm{b}}$ which implies that it is enough to consider the action of supercharges on $\psi_{\bm{0}}$ to study the degenerate states. 
The operator  $ X_{\bm{a}}$ in \eqref{MinimalX} is a product of gamma matrices, especially
\begin{equation}
 X_{10\dots 0} \sim \gamma_1, \quad
 X_{010\dots 0} \sim \gamma_2, \quad
 X_{0010\dots 0} \sim \gamma_3, \ \dots, \ 
 X_{0\dots 010} \sim \gamma_{n-1}.
\end{equation}
These are a set of generators of the Clifford algebra $Cl(n-1)$ (note that any $\tg{j}$ does not appear in $X_{\bm{a}}$). 
Note also that $ X_{0\dots 01} \sim \mathbb{I}_{2^{n-1}}. $ 
It follows that the set of supercharges $ \Lambda := \{ \; Q_{10\dots 0}, Q_{010\dots 0},\, \dots, \,  Q_{0\dots 010} \; \}$ together with $ Q_{0\dots 01}  $  generate all supercharges and their product. 
A product of an odd number of these supercharge gives an another supercharge and a product of an even number of them gives a central element. 
The number of independent operators generated by $\Lambda $ is $ 2^{n-1}. $ 
These operators produce, by the multiplication of $ Q_{0\dots 01}, $ another $ 2^{n-1}$ independent operators.  
Therefore total number of the linearly independent operators obtained as a product of supercharges is $ 2^n. $ 
These operators produce degenerate states when acting on $\psi_{\bm{0}}(x) $ so that the excited states are $ 2^n$ fold degenerate. 
\end{proof}


As is already pointed out, in this model of $\Zn$-graded SQM all the central elements in  each parity even subspace $ \g_{\bm{a}}$ are realized as an essentially the same operator. 
In order to have models (representations of $\g$) which have no or less degeneracy of the central elements, we will use the Clifford algebras of larger dimension in the subsequent sections.

%
\setcounter{equation}{0}
\section{$ 2^{n+1}$ dimensional $\Zn$-graded SQM by $Cl(2n)$} 
\label{SEC:next}

In this section a $ 2^{n+1}$ dimensional $\Zn$-graded SQM is constructed by the use of Clifford algebra $ Cl(2n).$ 
The gamma matrices \eqref{Clrep} of $ Cl(2n)$ are $2^n$ dimension. 
We write $ \bm{1} = (1,1,\dots,1) \in \Zn$ and we remark that $ Q_{\bm{1}}$ exists only 
when $n$ is odd.  

\begin{proposition} \label{P5}
Define the following $2^{n+1}$ dimensional matrix differential operators:
\begin{align}
   H &= \mathbb{I}_{2^n} \otimes \mathsf{H}, 
   \qquad
   Q_{\bm{a}} = Y_{\bm{a}} \otimes \mathsf{Q},
   \qquad
   Z_{\bm{a}\bm{b}} = (-i)^{1-\bm{a}\cdot \bm{b}} Y_{\bm{a}} Y_{\bm{b}} \otimes \mathsf{H},
   \label{2ndSmallest} \\
   Q_{\bm{1}} &= \mathbb{I}_{2^n} \otimes \mathsf{Q},
   \qquad
   Z_{\bm{a}\bm{1}} = Z_{\bm{1}\bm{a}} = Y_{\bm{a}} \otimes \mathsf{H}
   \label{DisappearEven}
\end{align}
where $ \bm{a}, \bm{b} \neq \bm{1}, \; \bm{a} \neq \bm{b} $ and 
\begin{equation}
   Y_{\bm{a}} = i^{h(\bm{a})} \gamma_1^{a_1} \gamma_2^{a_2} \cdots \gamma_n^{a_n},
   \qquad
   h(\bm{a}) = \sum_{j=1}^{n-1} \sum_{k=j+1}^n a_j a_k. 
   \label{Ydef}
\end{equation}
Then the matrix operators realize the relations \eqref{DefRels} of $\g.$ 
Thus they give a $2^{n+1}$ dimensional $\Zn$-graded SQM. 
\end{proposition}

Note that the operators in \eqref{DisappearEven} do not exist for even $n.$

\begin{proof}
The proof is rather simple. First, note that
\begin{equation}
   Y_{\bm{a}}^{\dagger} = Y_{\bm{a}}, \qquad (Y_{\bm{a}})^2 = \mathbb{I}_{2^n}.
\end{equation}
It follows immediately that $ Q_{\bm{a}}^2 = Q_{\bm{1}}^2 = \mathbb{I}_{2^n} \otimes \mathsf{H} = H.$ 
All $Y_{\bm{a}}$ is a product of odd number of $\gamma$'s. 
If $ \bm{a}\cdot \bm{b} = 0$ (resp. 1), then the number of gamma matrices common to $ Y_{\bm{a}}$ and $ Y_{\bm{b}} $ is even (resp. odd). 
Thus we have
\begin{equation}
     \bm{a}\cdot \bm{b} = 0 \ \Rightarrow \ \{ Y_{\bm{a}}, Y_{\bm{b}} \} = 0,
     \qquad
     \bm{a}\cdot \bm{b} = 1 \ \Rightarrow \ [ Y_{\bm{a}}, Y_{\bm{b}} ] = 0.
     \label{Ycom}
\end{equation}
It then follows 
\begin{equation}
  \llbracket Q_{\bm{a}}, Q_{\bm{b}} \rrbracket = \llbracket Y_{\bm{a}}, Y_{\bm{b}} \rrbracket \otimes \mathsf{H}
  = 2 Y_{\bm{a}} Y_{\bm{b}} \otimes \mathsf{H} = 2i^{1-\bm{a}\cdot\bm{b}} Z_{\bm{a}\bm{b}}.
\end{equation}
When $n$ is odd, because of $ \bm{a}\cdot \bm{1} = 1$
\begin{equation}
   \llbracket Q_{\bm{a}}, Q_{\bm{1}} \rrbracket = \{ Q_{\bm{a}}, Q_{\bm{1}} \} = 2Y_{\bm{a}} \otimes \mathsf{H}
   =2 Z_{\bm{a}\bm{1}}.
\end{equation}

Next, we show that  $ Z_{\bm{a}\bm{b}} $ and $ Z_{\bm{a}\bm{1}} $ have vanishing $\Zn$-graded bracket with all supercharges which implies the centrality of $Z$'s.  
The bracket with $ Q_{\bm{1}}$ is computed easily:
\begin{equation}
  \llbracket Z_{\bm{a}\bm{b}}, Q_{\bm{1}} \rrbracket = [ Z_{\bm{a}\bm{b}}, Q_{\bm{1}} ]  \sim [Y_{\bm{a}}Y_{\bm{b}}, \mathbb{I}_{2^n} ] \otimes \mathsf{H} \mathsf{Q} = 0
\end{equation}
and we see that $ \llbracket Z_{\bm{a}\bm{1}}, Q_{\bm{1}} \rrbracket = [ Z_{\bm{a}\bm{1}}, Q_{\bm{1}} ] = 0 $ in a similar way. 
The bracket with $ Q_{\bm{c}}, \; \bm{c} \neq \bm{1}$ is computed  as follows: 
It is also easy to see from \eqref{Ycom} that 
$ \llbracket Z_{\bm{a}\bm{1}}, Q_{\bm{b}} \rrbracket = \llbracket Y_{\bm{a}}, Y_{\bm{b}} \rrbracket \otimes \mathsf{HQ} = 0, $ since $ (\bm{a} + \bm{1}) \cdot \bm{b} = \bm{a} \cdot \bm{b} + 1$ implies that if $ (\bm{a} + \bm{1}) \cdot \bm{b} = 0$ (resp. 1), then $  \bm{a} \cdot \bm{b} = 1$ (resp. 0).  
For $ \bm{a},\bm{b} \neq \bm{1},$ 
$ \llbracket Z_{\bm{a}\bm{b}}, Q_{\bm{c}} \rrbracket = \llbracket Y_{\bm{a}} Y_{\bm{b}}, Y_{\bm{c}} \rrbracket \otimes \mathsf{HQ} = 0 $ is proved also by \eqref{Ycom}. 
When $ (\bm{a}+\bm{b})\cdot \bm{c} = 0, $ we have  $ \bm{a}\cdot\bm{c} = \bm{b}\cdot\bm{c} $ 
and $ \llbracket Y_{\bm{a}} Y_{\bm{b}}, Y_{\bm{c}} \rrbracket = [Y_{\bm{a}} Y_{\bm{b}}, Y_{\bm{c}}   ]. $ 
If $ \bm{a}\cdot\bm{c} = 0,$ then we expand the commutator into anticommutators  as \eqref{Com2Anti}. 
If $ \bm{a}\cdot\bm{c} = 1,$ then we employ the expansion like \eqref{Com2Com}. 
This shows that $ [Y_{\bm{a}} Y_{\bm{b}}, Y_{\bm{c}}   ] =0. $ 
When $ (\bm{a}+\bm{b})\cdot \bm{c} = 1, $ we have  $ \bm{a}\cdot\bm{c} \neq \bm{b}\cdot\bm{c} $ 
and $ \llbracket Y_{\bm{a}} Y_{\bm{b}}, Y_{\bm{c}} \rrbracket = \{ Y_{\bm{a}} Y_{\bm{b}}, Y_{\bm{c}} \}. $ 
Expansion such as \eqref{AntiExpand1} and \eqref{AntiExpand2} shows $ \{ Y_{\bm{a}} Y_{\bm{b}}, Y_{\bm{c}} \} =0. $ 
Thus the centrality of $ Z_{\bm{a}\bm{b}}$ and $ Z_{\bm{a}\bm{1}} $ has been proved and this completes the proof. 
\end{proof}

\begin{remark} \label{R2}
 For even $n$, each subspace of $\g$ spanned by the central elements consists only of the central elements in \eqref{2ndSmallest}, i.e., $ Z_{\bm{a}\bm{b}} $ with $\bm{a}, \bm{b} \neq \bm{1}.$ 
While, for odd $n,$ each subspace spanned by the central elements consist of the central elements in both \eqref{2ndSmallest} and \eqref{DisappearEven}, i.e., $ Z_{\bm{a}\bm{b}} $ and $Z_{\bm{a}\bm{1}}.$
\end{remark}

\begin{remark} \label{R3}
 By the same way as Remark \ref{R1}, one may see that if $ \bm{a}+\bm{b} = \bm{c}+\bm{d} $ and none of $ \bm{a}, \bm{b}, \bm{c}, \bm{d}$ are equal to $\bm{1},$ then $ Z_{\bm{a}\bm{b}} \sim Z_{\bm{c}\bm{d}}. $ 
While even if  $ \bm{a}+\bm{b} = \bm{c} + \bm{1} $ $ Z_{\bm{a}\bm{b}} $ and $ Z_{\bm{c}\bm{1}}$ are linearly independent, since  $ Z_{\bm{a}\bm{b}} \sim \gamma_1^{a_1+b_1} \cdots \gamma_n^{a_n+b_n} \otimes \mathsf{H} $ and $ Z_{\bm{c}\bm{1}} \sim \gamma_1^{c_1} \cdots \gamma_n^{c_n} \otimes \mathsf{H} $ and $ a_n + b_n = c_n+1. $ 
\end{remark}

From Remarks \ref{R2}, \ref{R3} we see the followings: 
in the model of $\Zn$-graded SQM of Proposition \ref{P5} with odd $n,$ the degeneracy of central elements in each subspaces is resolved partially. However, the situation is the same as Proposition \ref{P2} for even $n,$ namely, all central elements in each subspace are realized as an essentially same operator.

With this knowledge of the central element, we now investigate the spectrum of $H$ in \eqref{2ndSmallest}.  
The Hilbert space of this model is taken to be $ \mathscr{H} = L^2(\mathbb{R}) \otimes \mathbb{C}^{2^{n+1}}. $ 
We shall see that the representation of $\g$ in Proposition \ref{P5} is not irreducible for even $n$ but irreducible for odd $n.$ 
\begin{proposition} \label{P7}
The zero energy state of $H$ in \eqref{2ndSmallest} is  either non-existent or $ 2^n$ fold degenerate. 
While the excited states are $2^{n+1}$ fold degenerate. 
If $n$ is even (resp. odd) then the  eigenspace $ V_E$ for a fixed energy $E>0$ is a reducible (resp. irreducible) space as a representation space of $\g.$ 
When $n$ is even, the eigenspace $V_E$ is a direct sum of two irreducible spaces each of which is $2^n$ dimensional. 
\end{proposition}
\begin{proof}
The zero energy states are defined by $ Q_{\bm{a}} \Psi_0(x) = Q_{\bm{1}} \Psi_0(x) = 0. $ 
Thus one may repeat the discussion same as Proposition \ref{P3} to prove the first part of the proposition.

To discuss the excited states, we pick up two specific excited states of the same energy from the Hilbert space $\mathscr{H}.$ They are given by $2^{n+1}$ components vectors as follows:
\begin{align}
  & \psi_{\bm{0}}^{(1)}(x) = (\psi(x), 0, \dots, 0)^T,
  \qquad
  \psi_{\bm{0}}^{(2)}(x) = (0, \varphi(x), 0, \dots, 0)^T,
  \\
  &H \psi_{\bm{0}}^{(a)} = E \psi_{\bm{0}}^{(a)}, \quad E >0
\end{align}
We assume that these vectors have $\Zn$ degree $ \bm{0} = (0,0,\dots,0) $ for the sake of simplicity (this choice is not essential to prove the Proposition \ref{P7}). 
If there exists an operator $A_{12}$ mapping $ \psi_{\bm{0}}^{(1)} $ to $\psi_{\bm{0}}^{(2)}$, then $A_{12}$ is realized in $\mathscr{H}$ by a $ 2^{n+1} \times 2^{n+1}$ matrix having a non-vanishing entry at $(2,1)$ position.  
On the other hand, an operator $ A_{21}$ mapping $ \psi_{\bm{0}}^{(2)} $ to $\psi_{\bm{0}}^{(1)}$ should be realized as a $ 2^{n+1} \times 2^{n+1}$ matrix having a non-vanishing entry at $(1,2)$ position. 

First, let us show that if $n$ is even, then the realization of $\g$ given by Proposition \ref{P7} does not have operators corresponding to $ A_{12}$ and $ A_{21}.$ 
The operators $ A_{12}$ and $ A_{21} $ are, if any, given by  a tensor product of an element of $Cl(2n)$ and a $ 2 \times 2$ matrix operator which is diagonal $(\mathsf{H}) $ or antidiagonal $(\mathsf{Q}, \mathsf{HQ}).$
Therefore, the possibility is unique. 
$ A_{12}$ and $ A_{21}$ should be a tensor product of an element of $ Cl(2n)$ having non-vanishing entry at $(1,1)$ position and an antidiagonal matrix operator.  

Recall our representation of $ Cl(2n)$ in \eqref{Clrep} which is given by a $n$ times tensor product of Pauli matrices so that to have a non-vanishing  $(1,1)$ entry it must be a tensor product of $\sigma_3 $ and $\mathbb{I}_2,$ i.e., a diagonal matrix. 
The operator $Y_{\bm{a}}$ in \eqref{2ndSmallest} is a product of $\gamma$'s and $ \tilde{\gamma}$'s do not appear so that any product of $ Y_{\bm{a}}$ can not be diagonal. 
Thus any product of $ Q_{\bm{a}}, Z_{\bm{a}\bm{b}}$ in \eqref{2ndSmallest} is not able to be $ A_{12}$ nor $ A_{21}. $  This proves the non-existent of $ A_{12}, A_{21}$ for even $n$. 
On the other hand, for odd $n$ one may take $ Q_{\bm{1}}$ as $ A_{12} $ or $ A_{21}$ since it it a tensor product of the identity matrix and $ \mathsf{Q}. $ 

With this observation in our mind, let us next consider the excited states for even $n$ where we have operators only in \eqref{2ndSmallest}. 
One may repeat the discussion similar to Proposition \ref{P4}. 
From \eqref{2ndSmallest} $ Q_{\bm{a}} Q_{\bm{b}} \sim Z_{\bm{a}\bm{b}} $ so that it is enough to consider the action of supercharges. 
By definition of $ Y_{\bm{a}}$ in \eqref{Ydef} 
\begin{equation}
  Y_{10\dots 0} \sim \gamma_1, \quad 
  Y_{010\dots 0} \sim \gamma_2, \ \dots, \  
  Y_{0\dots 01} \sim \gamma_n  
\end{equation}
and these matrices generate the Clifford algebra $Cl(n).$ 
This implies that the set of supercharges $ \Lambda' := \{ \; Q_{10\dots 0}, Q_{010\dots 0}, \dots, Q_{0\dots 01} \; \} $ generates all supercharges and their product.  
The number of independent operators generated by the set $ \Lambda'$ is $ 2^n.$ 
Action of the $ 2^n$ operators on $ \psi_{\bm{0}}^{(1)} $ produces a set of degenerate states which span an eigenspace $ V_E^{(1)} $  of $H$ with the energy $E.$ 
The action of the $ 2^n$ operators  on $ \psi_{\bm{0}}^{(2)} $ produces a set of another degenerate states which span the another eigenspace $ V_E^{(2)}$ of $H$ with the same energy $E.$ 
Thus the eigenspace with the energy $E$ is given by $ V_E = V_E^{(1)} \otimes V_E^{(2)} \subset \mathscr{H} $ and $ \dim V_E = 2^{n+1}$ which gives the degeneracy of the excited states. 

As shown above there is no operator connecting  $ \psi_{\bm{0}}^{(1)} $ and $ \psi_{\bm{0}}^{(2)} $ so that  each eigenspace $ V_E^{(1)}, V_E^{(2)} $ is invariant under the action of $\g.$  
By  repeated action of supercharges an element belonging to any subspace $ \mathscr{H}_{\bm{a}}$ is obtained. Thus the space $V_E^{(1)} $ and $V_E^{(2)} $ are irreducible.

Finally, we look at the case of odd $n.$ 
This case is rather simple, since we have the operators in \eqref{DisappearEven} in addition to those in \eqref{2ndSmallest}. $ Q_{\bm{1}}$ connects $ \psi_{\bm{0}}^{(1)} $ and $ \psi_{\bm{0}}^{(2)} $ so that one may consider only one of them, say, $ \psi_{\bm{0}}^{(1)}.$ 
The set of supercharges $ \Lambda'$ together with $ Q_{\bm{1}}$ generates all supercharges and their product. 
The total number of independent operators generated by $ \Lambda'$ and $ Q_{\bm{1}}$ is $ 2^{n+1}.$ 
These $ 2^{n+1}$ operators act on $ \psi_{\bm{0}}^{(1)} $ and generate $ 2^{n+1}$ degenerate states. 
Obviously, the eigenspace of $H$ with a fixed energy created in this way is irreducible representation space of $\g.$ 
This completes the proof.
 \end{proof}


%

\setcounter{equation}{0}
\section{Maximal dimensional $\Zn$-graded SQM by $Cl(2^n-2)$}
\label{SEC:max}

In \S \ref{SEC:Mini} and \S \ref{SEC:next}, we provided the models of $\Zn$-graded SQM where the central elements having the same $\Zn$ degree are not linearly independent. 
We discuss, in this section, a model of   $\Zn$-graded SQM where all the central element sharing the common $\Zn$ degree is linearly independent. 
Such a model is achieved if any product of supercharges maps $ \psi_{\bm{0}}(x) $ to an inequivalent state. 
The simplest way to construct such a model would be the use of $ Cl(2^n-2)$ since the number of the supercharges is equal to the number of $\gamma$'s plus identity. 

To build a model, we introduce an ordering of supercharges. 
First, we make a choice of ordering for parity odd ($ p(\bm{a})=1$) elements $ \bm{a} \in \Zn. $   
One may choose, for instance, the lexicographic ordering, reverse lexicographic one and so on. 
In fact, choice of the ordering is irrelevant for model building.  
Let us denote the ordered parity 1 elements of $ \Zn$ as
\begin{equation}
   \bm{a}_1, \ \bm{a}_2, \ \dots, \ \bm{a}_{M-1}, \ \bm{a}_{M}, \qquad M := 2^{n-1} 
   \label{aorder}
\end{equation} 
One should not confuse $ \bm{a}_i $ with the $i$th component $a_i$ of an element $ \bm{a} $ of $\Zn.$ 
The ordering \eqref{aorder} is reflected in an ordering of supercharges. 

\begin{proposition} \label{P8}
  Define $2^{M}$ dimensional matrix differential operators by
  \begin{align}
     H &= \mathbb{I}_{2^{M-1}} \otimes \mathsf{H}, 
     \nn \\
     Q_k &:= Q_{\bm{a}_k} = G_k \otimes \mathsf{Q},
     \nn \\
     Z_{k\ell} &:= Z_{\bm{a}_k \bm{a}_{\ell}} = (-i)^{1-\bm{a}_k \cdot \bm{a}_{\ell}} 
     G_k G_{\ell} \otimes \mathsf{H}		
     \label{MaxRep}
  \end{align}
  where $ k, \ell \in \{ 1, 2, \dots, M \}$ and 
  \begin{equation}
     G_1 := \gamma_1, \qquad
     G_m := \prod_{j=1}^{m-1} (i\gamma_j \tg{j})^{\bm{a}_j \cdot \bm{a}_m} \gamma_m ,
     \qquad 
     G_M := \prod_{j=1}^{M-1} (i\gamma_j \tg{j})^{1-\bm{a}_j \cdot \bm{a}_M}
     \label{MaxG}
  \end{equation}
  with $ m = 2, 3, \dots, M-1.$ 
  Then the matrix differential operators realize the relations \eqref{DefRels} of $\g.$ 
  Thus they give a $2^{M}$ dimensional $\Zn$-graded SQM. 
\end{proposition}
\begin{proof}
The proof is almost same as Proposition \ref{P5}. 
 From \eqref{maxrel1} and \eqref{maxrel2} it is easy to see the relations
 \begin{equation}
     G_k^{\dagger} = G_k, \qquad G_k^2 = \mathbb{I}_{2^{M-1}}.
\end{equation}  
It follows immediately that $ Q_k^2 = H $ for any $k.$  

We have the relations similar to \eqref{Ycom}:
\begin{equation}
  \bm{a}_k \cdot \bm{a}_{\ell} = 0  \ \Rightarrow \  \{ G_k, G_{\ell} \} = 0,
  \qquad
  \bm{a}_k \cdot \bm{a}_{\ell} = 1  \ \Rightarrow \  [ G_k, G_{\ell} ] = 0. 
  \label{Gcom}
\end{equation}
To verify \eqref{Gcom}, one may assume that $ k < \ell $ without loss of generality. 
First, we consider  the case $ \bm{a}_k \cdot \bm{a}_{\ell} = 0. $
If $ \ell < M, $ then $ G_{\ell}$ does not have the factor $i \gamma_k \tg{k}.$ 
By definition, $ G_{\ell} $ does not have the factor $i\gamma_{\ell} \tg{\ell},$ either. 
It follows from $ \{ \gamma_k, \gamma_{\ell} \} = 0 $  that $\{G_k, G_{\ell} \} = 0.$ 
While, if $ \ell = M,$ then $ G_M $ has the factor $i \gamma_k \tg{k}.$ 
As $\gamma_k$ anticommutes with this factor, thus we have $\{G_k, G_M \} = 0.$ 
One may verify the case  $ \bm{a}_k \cdot \bm{a}_{\ell} = 1 $  in a similar way. 
Using \eqref{Gcom} the $\Zn$-graded bracket for supercharges given by \eqref{MaxRep} is computed as follows:   
\begin{equation}
 \llbracket Q_k, Q_{\ell} \rrbracket = \llbracket G_k, G_{\ell} \rrbracket \otimes \mathsf{H}  
  = 2 G_k G_{\ell} \otimes \mathsf{H}  
  = 2 i^{1-\bm{a}_k \cdot \bm{a}_{\ell}} Z_{k \ell}.
\end{equation}

The  centrality of $ Z_{k\ell} $ stems also from \eqref{Gcom}.   
By \eqref{MaxRep}, we have 
$ \llbracket Z_{k\ell}, Q_m \rrbracket\sim \llbracket G_k G_{\ell}, G_m \rrbracket \otimes \mathsf{H} \mathsf{Q}. $ 
The identity $ \llbracket G_k G_{\ell}, G_m \rrbracket = 0 $ is easily seen, by \eqref{Gcom},  in the same way as Proposition \ref{P5}. 
This completes the proof. 
\end{proof}

\begin{remark} \label{R4}
 All the central elements $ Z_{k\ell}$ are linearly independent. 
 This is readily seen noticing that $ Z_{k\ell} \sim Q_k Q_{\ell} \sim G_k G_{\ell} \otimes \mathsf{H}. $ 
Since $ i \gamma_j \tg{j}$ is a diagonal matrix in the representation \eqref{Clrep}, the position of non-zero entries in $ G_k G_{\ell} $ is determined by $ \gamma_k \gamma_{\ell}. $  
Thus  $ Z_{k\ell}$ can not share the position of non-zero entries with  $  Z_{ts} $ so that they are linearly independent.  
\end{remark}

 Now let us study the spectrum of $H$ in \eqref{MaxRep}. 
The Hilbert space of this model is taken to be $ \mathscr{H} = L^2(\mathbb{R}) \otimes \mathbb{C}^{2^M}.$  
The ground state $ \Psi_0(x)$ is defined by $ Q_{\bm{a}} \Psi_0(x) =0 $ for all possible $\bm{a}. $ 
Thus one may repeat the same discussion as Proposition \ref{P3} and see that the ground states are, if any, $2^{M-1}$ fold degenerate.  
The degeneracy of the excited states can be understood in a way similar to Propositions \ref{P4} and \ref{P7}. 
Namely, we consider the algebra whose elements are a product of supercharges. 
In the realization \eqref{MaxRep}, the set of all supercharges is a generator of the algebra (compare with Proposition \ref{P4} and \ref{P7} where a subset of supercharges generates the corresponding algebra). 
This is seen by noting that $ G_k \; (1 \leq k \leq M-1) $ is a product of a diagonal matrix and $\gamma_k,$ while $ G_M $ is a diagonal matrix. 
This implies that the set $ \{ \; G_1,G_2, \dots, G_{M-1} \; \} $ can be a set of generators of an algebra which is similar to  $Cl(M-1).$ 
The difference from $Cl(M-1)$ is the presence of commuting pairs of $G$'s.

Thus any supercharge should be taken as a generator. 
The number of supercharges is $ M $ so that the algebra generated by all supercharges is $ 2^M.$ 
Let $\psi_{\bm{0}}(x) = (\psi(x), 0, \dots, 0)^T \in \mathscr{H}_{\bm{0}} $ be an excited state wave function. 
Then the $ 2^M$ states
\begin{equation}
  \psi_{\bm{0}}, \quad Q_k \psi_{\bm{0}}, \quad Q_k Q_{\ell} \psi_{\bm{0}}, \quad Q_k Q_{\ell} Q_m \psi_{\bm{0}}, \quad \dots, \quad Q_1 Q_2 \cdots Q_M  \psi_{\bm{0}}
\end{equation}
exhaust all the linearly independent degenerate states obtained from $\psi_{\bm{0}}$ by  sequential action of the supercharges. 

Therefore, we have established the following:
\begin{proposition}
The zero energy state of $H$ in \eqref{MaxRep} is  either non-existent or $ 2^{M-1}$ fold degenerate. 
While the excited states are $ 2^M$ fold degenerate.  
\end{proposition}

%
\setcounter{equation}{0}
\section{Models in between minimal and maximal dimensions}
\label{SEC:InBet}

In the previous sections, we provided three models of $ \Zn$-graded SQM. 
Depending the choice of the Clifford algebra $Cl(2m)$, the model has minimal $(2^n)$, next to the minimal $(2^{n+1})$ and maximal  $ (2^{2^{n-1}})$ dimensions. 
In Table \ref{TAB:Cl}, we indicate the Clifford algebras used in the previous sections   for construction of the models  for some small values of  $n.$

\begin{table}[h]
\begin{center}
\begin{tabular}{c|ccccc}
  $n$ & minimal & next & & &  maximal \\ \hline
   2 & $ Cl(2)$  & & & &  $Cl(2)$ \\[3pt]
   3 & $ Cl(4)$  & & & &  $Cl(6)$ \\[3pt]
   4 & $ Cl(6)$ & $Cl(8) $ & $ Cl(10)$ & $Cl(12)$ & $ Cl(14) $ \\[3pt]
   5 & $ Cl(8)$ & $Cl(10)$ & $\cdots $ & $Cl(28)$ & $ Cl(30)$ 
\end{tabular}
\end{center}
\caption{Clifford algebras used for the models} \label{TAB:Cl}
\end{table}

When $n=2, $ there is no distinction between models of the minimal and the maximal dimensions. 
For $n=3,$ there is no more Clifford algebras can  be used for model building between the minimal and the maximal dimensions. 
From $n=4, $ we see the existence of more Clifford algebras which have not been used in the models of the  minimal, next to the minimal and the maximal dimensions. 
The constructions in the previous sections suggest strongly that these Clifford algebras may be used to build  models of  various dimensions in between the minimal and the maximal. 
Proving it for arbitrary $n$ is highly non-trivial problem and the problem is still open. 
Here we prove it for $n=4$ by presenting such models explicitly. 
We also present two more models for $n=5$ in order to support the anticipation.

 The models presented in this section have the common structure. 
We order the parity 1 elements of $\Zn$ lexicographically:

\noindent
For $n = 4; $ 
\begin{alignat}{4}
 \bm{a}_1 &= (0,0,0,1), &\quad  \bm{a}_2 &=(0,0,1,0), & \quad \bm{a}_3 &= (0,1,0,0), & \quad \bm{a}_4 &= (1,0,0,0),
 \nn \\
 \bm{a}_5 &=(0,1,1,1), & \bm{a}_6 &= (1,0,1,1),  & \bm{a}_7 &= (1,1,0,1), & \bm{a}_8 &= (1,1,1,0).
\end{alignat}

\noindent
For $n = 5; $ 
\begin{alignat}{4}
 \bm{a}_1 &= (0,0,0,0,1), &\quad  \bm{a}_2 &=(0,0,0,1,0), & \quad \bm{a}_3 &= (0,0,1,0,0), & \quad \bm{a}_4 &= (0,1,0,0,0),
 \nn \\
 \bm{a}_5 &=(1,0,0,0,0), & \bm{a}_6 &= (0,0,1,1,1),  & \bm{a}_7 &= (0,1,0,1,1), & \bm{a}_8 &= (1,0,0,1,1),
 \nn \\
 \bm{a}_9 &=(0,1,1,0,1), & \bm{a}_{10} &=(1,0,1,0,1), & \bm{a}_{11} &=(1,1,0,0,1), & \bm{a}_{12} &=(0,1,1,1,0),
 \nn \\
 \bm{a}_{13} &=(1,0,1,1,0), & \bm{a}_{14} &=(1,1,0,1,0), & \bm{a}_{15} &=(1,1,1,0,0), & \bm{a}_{16} &=(1,1,1,1,1).
\end{alignat}
We consider the  realization of $\Zn$-graded Poincar\'e algebra by $ Cl(2m)$ of the following form:
\begin{align}
   H &= \mathbb{I}_{2^m} \otimes \mathsf{H}, \qquad Q_k := Q_{\bm{a}_k} = G_k \otimes \mathsf{Q},
   \nn\\
   Z_{k\ell} &:= Z_{\bm{a}_k \bm{a}_{\ell}} = (-i)^{1-\bm{a}_k \cdot \bm{a}_{\ell}} G_k G_{\ell} \otimes \mathsf{H}
\end{align}
where $ G_k \in Cl(2m).$ 

Now we present models of $\mathbb{Z}_2^4$-graded SQM by $Cl(12)$ and $Cl(10).$ 
\begin{proposition} \label{P11}
 A $2^7$ dimensional model of $ \mathbb{Z}_2^4$-graded SQM is given by the following choice of $ G_k \in Cl(12):$
 \begin{alignat}{4}
   G_1 &= \gamma_1, & \quad G_2 &= \gamma_2, & \quad G_3 &= \gamma_3, & \quad G_4 &= \gamma_4,
   \nn \\
   G_5 &= \Gamma_1 \Gamma_2 \Gamma_3 \gamma_5, & G_6 &= \Gamma_1 \Gamma_2\Gamma_4 \gamma_6, 
   & G_7 &=\Gamma_2 \Gamma_5 \Gamma_6, & G_8 &= \gamma_1 \tg{2} \tg{3} \tg{4}
   \label{Z24Cl12}
 \end{alignat}
 where $ \Gamma_k = i \gamma_k \tg{k}.$
 
 A $ 2^6$ dimensional  model of $\mathbb{Z}_2^4$-graded SQM is given by the $ G_k \in Cl(10)$ which has the same form as \eqref{Z24Cl12} for $k=1,2,\dots,5$ while
 \begin{equation}
       G_6 = \Gamma_3 \Gamma_5, \qquad 
       G_7 = \tg{1} \gamma_2 \tg{3} \tg{4}, \qquad
       G_8 = i \gamma_2 \gamma_3 \gamma_4.
 \end{equation}
\end{proposition}
\begin{proof}
 Since the realization \eqref{Z24Cl12} has the same form as Proposition \ref{P8}, the present proposition can be proved in the same way as Proposition \ref{P8}. 
Namely, it is enough to verify the relations \eqref{Gcom}. 
Note that all $G$'s are taken to be hermitian and idempotent. 
Observe that for $ Cl(12)$ (resp. $ Cl(10)$), from $ G_1 $ to $ G_7 $ (resp. $G_6$) are exactly same form as \eqref{MaxG}. 
Therefore \eqref{Gcom} has been verified for these $G$'s. 
By direct computation, one may easily verify \eqref{Gcom} for other $G$'s so we do not present the computational detail. 
\end{proof}

Proposition \ref{P11} completes the construction of a series of models of $\mathbb{Z}_2^4$-graded SQM in Table \ref{TAB:Cl}. 
Thus we have proved the existence of models of $\mathbb{Z}_2^4$-graded SQM from the minimal to the maximal dimensions. 

In the $ Cl(10) $ model, $ G_3 G_4 = \gamma_3 \gamma_4$ and $ G_2 G_8 = i \gamma_3\gamma_4$ so that $ Z_{34} $ and $ Z_{28}$ are linearly dependent. 
On the other hand, one may verify from expressing $ G_k G_{\ell}$ in terms of $\gamma$'s that all $Z$'s is linearly independent in $ Cl(12)$ model. 
Thus the $ Cl(12)$ model is the \textit{minimal} model in which all the central elements are linearly independent.

Let us turn to $\mathbb{Z}_2^5$-graded SQM by $ Cl(28) $ and $ Cl(26).$ 

\begin{proposition} \label{P12}
  A $ 2^{15}$ dimensional  model of $\mathbb{Z}_2^5$-graded SQM is given if $ G_k \in Cl(28)$ are taken as follows:
\begin{alignat}{3}
  G_j &= \gamma_j \quad (j=1,\dots,5), & \quad 
  G_6 &= \Gamma_1 \Gamma_2 \Gamma_3 \gamma_6, & \quad 
  G_7 &= \Gamma_1 \Gamma_2 \Gamma_4 \gamma_7,
  \nn \\
  G_8 &=  \Gamma_1 \Gamma_2 \Gamma_5 \gamma_8, & \quad
  G_9 &= \Gamma_1 \Gamma_3 \Gamma_4 \Gamma_8 \gamma_9, & \quad
  G_{10} &= \Gamma_1 \Gamma_3 \Gamma_5 \Gamma_7 \gamma_{10},
  \nn \\
  G_{11} &= \Gamma_1 \Gamma_4 \Gamma_5 \Gamma_6 \gamma_{11}, & \quad
  G_{12} &= \Gamma_2 \Gamma_3 \Gamma_4 \Gamma_8  \Gamma_{10} \Gamma_{11} \gamma_{12}, & \quad
  G_{13} &= \Gamma_2 \Gamma_3 \Gamma_5 \Gamma_7  \Gamma_9 \Gamma_{11} \gamma_{13},
  \nn \\
  G_{14} &= \Gamma_2 \Gamma_4 \Gamma_5 \Gamma_6  \Gamma_9 \Gamma_{10} \gamma_{14}, & \quad 
  G_{15} &= \Gamma_1 \Gamma_2 \prod_{j=9}^{14} \Gamma_j, & \quad
  G_{16} &= \prod_{j=1}^5 \tg{j} \prod_{k=6}^8 \gamma_k.
  \label{Z25Cl28}   
\end{alignat}  
  A $ 2^{14}$ dimensional  model of $\mathbb{Z}_2^5$-graded SQM is given if $ G_k \in Cl(26)$ are taken as follows: from $ G_1 $ to $ G_{13} $ and $ G_{16}$ are exactly same as \eqref{Z25Cl28}, while
  \begin{equation}
       G_{14} = \Gamma_1 \Gamma_3 \Gamma_7 \Gamma_8 \Gamma_{11} \Gamma_{12} \Gamma_{13},
       \qquad 
       G_{15} = \gamma_3 \gamma_4 \gamma_5 \tg{12} \tg{13}.
  \end{equation}
\end{proposition}

\noindent
Note that $ G_1, G_2, \dots, G_{15} $ in \eqref{Z25Cl28} have the same form as \eqref{MaxG}. 
All $G$'s in Proposition \ref{P12} are taken to be hermitian and idempotent. 
\begin{proof}
  The proposition is proved by verifying \eqref{Gcom} by direct computation.  
\end{proof}

In the models of Proposition \ref{P12} all central elements are linearly independent. 
The $G$'s in Proposition \ref{P12} are based on the ones in Proposition \ref{P8} where $ \gamma_k$ is assigned to $ G_k $ except the $G_M. $ 
However, in Proposition \ref{P12} we have a smaller number of gamma matrices. 
If we try to repeat the construction similar to Proposition \ref{P12} with the Clifford algebras of smaller dimension, the realization of $\g$ would be more difficult as we have a smaller number of gamma matrices.

%
\section{Concluding remarks} \label{SEC:CR}

We presented several models of $ \Zn$-graded SQM. 
All the models presented share the common structure: a tensor product of the operators $ \mathsf{H}, \mathsf{Q}$ of the standard SQM and the Clifford gamma matrices of various dimensions. 
Up to $n=4$ we have shown the existence of a sequence of models from the minimal to the maximal dimensions. 
For $ n \geq 5 $ it will be also true the existence of a sequence of models. 
Considering of the models of various dimensions is made legitimate by linear independence of the central elements of $\g.$ 
Use of the Clifford algebras of small dimension produce $\Zn$-graded SQMs where some of the central elements are realized as linearly dependent operators, while the large dimensional Clifford algebras do $\Zn$-graded SQMs where all the central elements are linearly independent.

Although we have proved that the $ {\cal N}=1$ SQM is possible to extend to a $\Zn$-graded setting, physical meanings of $ \Zn$-grading is unclear.  
We observed the existence of a sequence of models of $\Zn$-graded SQM for a fixed value of $n.$ Physical meaning of this sequence is also unclear. 
One way to understand the $ \Zn$-graded SQM further would be considering a $\Zn$-graded classical mechanics. 
That is, try to find classical actions which are invariant under $\Zn$-graded supersymmetry transformation, then quantize the theory to get the $\Zn$-graded SQMs of the present work. 
Especially, the counterpart of the sequence of $\Zn$-graded SQM for a fixed $n$ in classical mechanics will provide interesting information about $\Zn$-grading structure. 
To develop such $\Zn$-graded classical mechanics the mathematical theory of $\Zn$-graded supermanifolds developed in \cite{Mar,CGP3,CGP4,CoKPo} will be helpful. 
Another way to understand meanings of $\Zn$-grading is to consider multipartite systems. 
In the present work, we restrict ourselves to a single particle system (see  \eqref{StandardSQM}). 
Multiparticle extensions will be an urgent task for further developments.

Another interesting direction to study is  $ \Zn$-graded SQM with ${\cal N} > 1$ and  conformal extensions. In the case of $\mathbb{Z}_2^2$-graded SQM, these two extensions of ${\cal N}=1 $ $\mathbb{Z}_2^2$-graded SQM were done in  \cite{NAKASD} where the conformal extension was also given by taking a tensor product of the Clifford gamma matrices and a model of standard superconformal mechanics.  
It is an interesting question whether similar construction is possible for $n >2.$ 

Finally, we comment a universality of Propositions \ref{P2}, \ref{P5}, \ref{P8}, \ref{P11},  \ref{P12}. 
We assume in these propositions that $ \mathsf{Q}, \mathsf{H}$ are given by \eqref{StandardSQM} as we are interested in quantum mechanics. 
However, this assumption is not needed. $ \mathsf{Q}, \mathsf{H}$ can be any representations (not necessary quantum mechanical ones) of the standard supersymmety algebra. 
If we take a representation of $\mathsf{Q}, \mathsf{H}$ different from \eqref{StandardSQM}, then $ \sigma_3$ in Proposition \ref{P2} should be replaced with $ \mathsf{S} $ such that $ \{ \mathsf{Q}, \mathsf{S} \} = 0 $ since the proof of the propositions are independent of the explicit form of $ \mathsf{Q}, \mathsf{H}. $

%
%
%
%
%
%

\end{document}